\documentclass[twoside]{LNGAI}

\usepackage{LNGAImacro}

\usepackage[utf8]{inputenc}
\usepackage{graphicx}
\usepackage{amsmath}
\usepackage{amssymb}
\usepackage{xspace}
\usepackage{fancybox}
\usepackage{enumitem}
\usepackage{makecell}

\newcommand{\reservedWordTiles}[1]{\ovalbox{\ensuremath{\mathsf{#1}}\xspace}}
\newcommand{\tiles}[1]{\reservedWordTiles{#1}\xspace}

\newcommand{\Soda}{\textsc{Soda}\xspace}
\newcommand{\Tiles}{\ensuremath{\mathsf{Tiles}}\xspace}
\newcommand{\AR}{\ensuremath{\mathsf{AR}}\xspace}

\newcommand{\Acrocpolis}{ACROCPoLis\xspace}
\newcommand{\Acromagat}{AcROMAgAt\xspace}

\newcommand{\fFair}{\ensuremath{\tilde{f_{\mathsf{F}}}}\xspace}
\newcommand{\fFairTauSeqi}{\ensuremath{\tilde{f}_{\tau(\mathsf{F},O)seqi}}\xspace}
\newcommand{\fPowerSet}{\ensuremath{\mathcal{P}}\xspace}
\newcommand{\fOrderedSeq}[1]{\ensuremath{{#1}_{\prec}}\xspace}
\newcommand{\fPerm}{\ensuremath{\mathrm{Perm}}}

\newcommand{\sAgent}{\ensuremath{\mathsf{A}}\xspace}
\newcommand{\sAgentAttribute}{\ensuremath{\mathsf{A_{at}}}\xspace}
\newcommand{\sResource}{\ensuremath{\mathsf{R}}\xspace}
\newcommand{\sResourceAttribute}{\ensuremath{\mathsf{R_{at}}}\xspace}
\newcommand{\sPSOutcome}{\ensuremath{\fPowerSet(\mathsf{\sAgent \times \sResource})}\xspace}
\newcommand{\tFairScen}{\ensuremath{\mathsf{F}}\xspace}
\newcommand{\sMeasure}{\ensuremath{\mathsf{M}}\xspace}
\newcommand{\sZeroOne}{\ensuremath{\{0, 1\}\xspace}}
\newcommand{\sZeroOneInterval}{\ensuremath{[0, 1]\xspace}}
\newcommand{\sBoolean}{\ensuremath{\mathbb{B}}\xspace}
\newcommand{\sNat}{\ensuremath{\mathbb{N}_{0}}\xspace}
\newcommand{\sRational}{\ensuremath{\mathbb{Q}}\xspace}
\newcommand{\sReal}{\ensuremath{\mathbb{R}}\xspace}
\newcommand{\sInteger}{\ensuremath{\mathbb{Z}}\xspace}

\newcommand{\tilestype}[1]{{\ensuremath {\textsf{#1}}}}
\newcommand{\tilesfun}[1]{{\ensuremath {\textsf{#1}}}}

\newcommand{\orcidID}[1]{{ORCID: #1}}
\newcommand{\email}[1]{\texttt{#1}}

\newtheorem{notation}{Notation}[section]

\newtheorem{concept}{Concept}[section]

\def\lastname{Mendez and Kampik}

\begin{document}

    \begin{frontmatter}
        \title{Specification, Application, and Operationalization of a Metamodel of Fairness}

        \author{Julian~Alfredo~Mendez}\footnote{\orcidID{0000-0002-7383-0529}, \email{julian.mendez@cs.umu.se}}
        \address{Umeå University \\ Sweden}
        \author{Timotheus Kampik}\footnote{\orcidID{0000-0002-6458-2252}, \email{tkampik@cs.umu.se}}
        \address{Umeå University \\ Sweden}

        \begin{abstract}
            This paper presents the \AR fairness metamodel, aimed at formally representing, analyzing, and comparing fairness scenarios. The metamodel provides an abstract representation of fairness, enabling the formal definition of fairness notions. We instantiate the metamodel through several examples, with a particular focus on comparing the notions of equity and equality.
            We use the \Tiles framework, which offers modular components that can be interconnected to represent various definitions of fairness. Its primary objective is to support the operationalization of \AR-based fairness definitions in a range of scenarios, providing a robust method for defining, comparing, and evaluating fairness.
            \Tiles has an open-source implementation for fairness modeling and evaluation.
        \end{abstract}

    \end{frontmatter}

    \section{Introduction}
    \label{sec:intro}

    Fairness is a critical consideration in various domains, including social policy, economics, and technology. Despite its importance, defining and evaluating fairness remains a complex challenge, as assessments of what is fair may vary across contexts and stakeholders. There is no universal definition of fairness, and even seemingly purely technical decisions can have direct fairness implications~\cite{AlerTubella-2022}. Given a specific scenario, a definition of fairness can be addressed by defining a \emph{fairness measure} that measures the fair distribution of resources among agents. Although fairness measures are subjective, they must be well defined in critical contexts. This makes it essential to systematize how these measures are defined, from an abstract subjective understanding to concrete execution, and to support comparing fairness definitions.

    This paper introduces the \AR fairness metamodel, designed to represent and analyze fairness measures and scenarios. The metamodel extends the previous research in~\cite{Mendez.Kampik.Aler.Dignum-2024-SCAI} and serves as a model of models~\cite{Weske-2019}, where each model is an instance of the metamodel.
    The specific instances of these models can then be evaluated to verify whether they comply with a given definition of fairness. The metamodel addresses the challenges of defining and evaluating fairness by offering a structured approach. It provides an abstract representation of \emph{fairness scenarios}, incorporating key elements such as agents, resources, and attributes, which are essential components for evaluating whether a given outcome adheres to a specific definition of fairness.

    We use \Tiles~\cite{Mendez.Kampik.Aler.Dignum-2024-SCAI}, a framework designed to support the \AR fairness metamodel. \Tiles consists of modular blocks, called \emph{tiles}, that can be interconnected to specify a definition of fairness. Each block is annotated to indicate how it can be connected to other blocks within the framework. The combination of the \Tiles framework and the \AR fairness metamodel provides a comprehensive set of tools to model fairness and evaluate fairness in various scenarios, and can be applied to real-world situations.

    This paper is organized as follows. Section~\ref{sec:background} provides an overview of computational models of fairness. Section~\ref{sec:metamodel} introduces the \AR fairness metamodel and its components, including the definition of identifiers, measures, attributes, and auxiliary functions.
    Section~\ref{sec:operationalization} discusses the structure of the blocks and their graphical notation. Section~\ref{sec:discussion} offers a discussion of the \AR fairness metamodel and the \Tiles framework, focusing on their capabilities and limitations. Finally, Section~\ref{sec:conclusion} concludes with reflections and a discussion of future work.

    \section{Background}
    \label{sec:background}

    The importance of fairness in machine learning and artificial intelligence (AI) systems is widely recognized.
    From a modeling perspective, evaluating fairness requires the ability to identify and quantify unwanted bias, which may lead to prejudice and ultimately to discrimination.

    Formalizing fairness can lead to greater transparency in achieving equitable outcomes, which benefits both individuals and the groups they represent.
    Although operationalizing fairness is challenging, efforts to formalize it and automate fairness verification~\cite{Albarghouthi-2017,Albarghouthi-2019} are relevant.
    Several quantifiable definitions have been proposed~\cite{Dwork-2012,Hardt-2016,Joseph-2016,Kearns-2018}, reflecting legal, philosophical, and social perspectives.
    However, different interpretations can inadvertently harm the groups they aim to protect~\cite{CorbettDavies-2018} or do not account for intersectionality~\cite{Kearns-2018}.

    Two widely discussed formalizations are \emph{individual fairness} and \emph{group fairness}.
    Individual fairness requires that similar individuals, based on non-protected attributes, receive similar outcomes.
    Group fairness stipulates that protected groups should receive similar outcomes when non-protected factors are equal~\cite{Chouldechova-2017}.
    These notions can conflict~\cite{Binns-2019-IndividualGroupFairness}.
    For example, if two individuals with similar qualifications receive different outcomes solely because they belong to different protected groups, then group fairness metrics such as equality of odds or equality of opportunity can be used to address the disparity.
    In practice, reconciling these notions and managing the associated value trade-offs remains an active research challenge~\cite{DBLP:conf/innovations/DworkHPRZ12,10.1145/3461702.3462621,AlerTubella-2022}. Model-based methodologies, such as MBFair~\cite{Ramadan-2025-SSM}, enable the verification of software designs with respect to individual fairness.

    Operational tools for fairness assessment include IBM's AI Fairness 360~\cite{DBLP:journals/ibmrd/BellamyDHHHKLMM19} (AIF360), Microsoft's Fairlearn~\cite{bird2020fairlearn}, and Google's What-if Tool~\cite{DBLP:journals/tvcg/WexlerPBWVW20} (WIT).
    AIF360 is a comprehensive and technical open-source Python library that includes fairness metrics and bias mitigation algorithms. It works in the three stages: pre-processing, in-processing, and post-processing.
    Fairlearn includes fairness metrics, bias mitigation algorithms, and also provides fairness dashboards for visual comparisons. WIT is visualization-oriented and provides a dashboard to explore counterfactuals to answer the question ``What if this feature changed?''. However, these tools focus on the operationalization of fairness measures rather than their definition and analysis. To address this limitation, we propose a unified metamodel that supports various perspectives of fairness, building on the frameworks \Acrocpolis ~\cite{AlerTubella-2023} and \Acromagat ~\cite{Mendez.Kampik.Aler.Dignum-2024-SCAI}. We aim to integrate multiple definitions of fairness into a coherent structure, facilitating consistent evaluation and comparison between scenarios.

    \section{Fairness Metamodel}
    \label{sec:metamodel}

    This section presents \AR, a formal metamodel from which fairness definitions can be instantiated.
    Conceptually, \AR focuses on \emph{agents}, \emph{resources}, their attributes as first-class abstractions, and the \emph{outcomes}.
    The presentation of \AR is accompanied by several examples that demonstrate its applicability to the assessment of fairness in specific scenarios, as well as the more abstract comparison of fairness measures.

    \subsection{Basic Elements}
    \label{subsec:basic}

    A fairness scenario provides the building blocks for relating agents, resources, and their attributes. This relation, called \emph{outcome}, is used to evaluate whether it adheres to a defined concept of fairness.
    As a prerequisite, we assume two finite background sets, one of (not further specified) \emph{agents}, denoted by $\cal A$, and one of (not further specified) \emph{resources}, denoted by $\cal R$.
    We assume that the two sets are disjoint, i.e., ${\cal A} \cap {\cal R} = \emptyset$.
    The metamodel is defined as follows.
    \begin{definition}[Fairness Scenario]
        \label{def:metamodel}
        A \emph{fairness scenario} is a tuple $\tFairScen = \langle \sAgent, \sResource, \sAgentAttribute, \sResourceAttribute \rangle$, such that:
        \begin{itemize}
            \item $\sAgent \subseteq \cal A$;
            $\sResource \subseteq \cal R$;
            $\sAgent$ and $\sResource$ are non-empty;
            \item every \emph{agent attribute} $\mathsf{a_{AT}} \in \sAgentAttribute$ is a function that takes an agent as input; every \emph{resource attribute} $\mathsf{r_{AT}} \in \sResourceAttribute$ is a function that takes a resource as input; the codomains of agent attributes and resource attributes may vary and are specified upon instantiation.
        \end{itemize}
    \end{definition}

    Note that in our metamodel, we exclude functions that operate on multiple agents, multiple resources, or combinations of them.

    A consistent definition of quantities is crucial for measuring fairness.
    \begin{notation}[Sets of quantities]
        The set $\sMeasure$ is a placeholder for a set of quantities such as the set of real numbers ($\sReal$), rational numbers ($\sRational$), integers ($\sInteger$), or natural numbers including 0 ($\sNat$), with its operations totally defined on $\sMeasure$.
    \end{notation}

    Relevant attributes are, for example:
    \begin{enumerate}[label=\roman*)]
        \item the utility function $u : \sResource \to \sMeasure$, which returns how much a resource is worth, and
        \item the need function $q : \sAgent \to \sMeasure$, which returns how much of a resource is needed by an agent.
    \end{enumerate}
    Given a fairness scenario, we can define fairness measures. Observe that we denote the \emph{power set} of a set $S$ by $\fPowerSet(S)$.
    \begin{definition}[Fairness Measure]
        \label{def:fairness-measure}
        Let $\tFairScen = \langle \sAgent, \sResource, \sAgentAttribute, \sResourceAttribute \rangle$ be a fairness scenario.
        An \emph{outcome} $O$ is an element $O \in \sPSOutcome$, and when the outcome $O$ is clear from context, we may say that $a$ \emph{receives} $b$ whenever $\langle a, b \rangle \in O$. A \emph{fairness measure} $\fFair$ with respect to a fairness scenario $\tFairScen$ is a function $\fFair : \sPSOutcome \to \sZeroOneInterval$.

        Since $\sZeroOne$ is isomorphic to $\sBoolean$, which is $\{ false, true \}$, we especially consider the case when $\fFair (O)$ returns only 0 or 1. Our interpretation is that 0 corresponds to $false$ and 1 to $true$, i.e. if $\fFair (O) = 1$, the outcome is fair, and if $\fFair (O) = 0$, it is unfair. Generally, we use $\{ false, true \}$ and $\{ 0, 1 \}$ interchangeably.
    \end{definition}
    Let us take a look at how the definition of a fairness measure can be applied.
    \begin{notation}[Notation of Functions by Extension]
        \label{def:functions_by_extension}
        For a function $f : A \to B$, we denote $f$ as a set of pairs $\langle x , y \rangle$ such that $x$ ranges on the elements of $A$ exactly once and $y = f(x)$.
        We also use the Iverson bracket notation, where $ f(x) = [P(x)] $ denotes that $f(x) = 1$ if $P(X)$, and $f(x) = 0 $ otherwise.
    \end{notation}
    The example below illustrates how our fairness metamodel can be instantiated.
    \begin{example}[Fairness Scenario]
        A group of agents $\sAgent$---namely Alice ($A$), Bob ($B$), Carol ($C$), David ($D$), Eve ($E$), and Frank ($F$)---apply for a subsidy.
        Let us assume that there are three types of resources $R_{1}$, $R_{2}$, and $R_{3}$, and their utility $u$ is 10, 20, and 30 respectively.
        The agents needs are encoded in the function $q$, where $A$ and $D$ need 10, $B$ and $E$ need 20, and $C$ and $F$ need 30 each.
        Let us assume that it is \emph{fair} to give everyone at least one of the two best resources. The full instantiation of the fairness scenario and fairness measure is summarized as follows:

        \begin{itemize}
            \renewcommand\labelitemi{}
            \item $\sAgent$ = $\{A, B, C, D, E, F\}$,
            $\sResource$ = $\{ R_{1}, R_{2}, R_{3} \}$,
            \item $\sAgentAttribute$ = $\{ q : \sAgent \to \sNat$,
            $q = \{ \langle A, 10 \rangle, $ $ \langle B, 20 \rangle, $ $ \langle C, 30 \rangle, $ $ \langle D, 10 \rangle, $ $ \langle E, 20 \rangle, $ $ \langle F, 30 \rangle \} \ \}$
            \item $\sResourceAttribute$ = $\{ u : \sResource \to \sNat$, $u = \{ \langle R_{1}, 10 \rangle, $ $ \langle R_{2}, 20 \rangle, $ $ \langle R_{3}, 30 \rangle \} \ \}$,
            \item $\fFair (O)$ = $[ \forall a \in \sAgent \ (a \text{ receives } R_{2} \text{ or } a \text{ receives } R_{3}) ]$
        \end{itemize}

        $\sAgent$ and $\sResource$ contain the agents and resources respectively, $q$ specifies how much each agent needs, and $u$ specifies the utility of each resource.
        An intuition of $\fFair (O)$ is that every agent receives $R_{2}$ or $R_{3}$ (or both).

        Considering two different outcomes:
        $O_{1}= \{ \langle A, R_{3} \rangle,$ $ \langle B, R_{3}\rangle , $ $\langle C, R_{3} \rangle, $ $\langle D, R_{3} \rangle , $ $\langle E, R_{3} \rangle ,$ $\langle F, R_{3} \rangle \}$ and
        $O_{2}= \{ \langle A, R_{3} \rangle, $ $\langle B, R_{2}\rangle ,$ $\langle C, R_{1} \rangle,$ $\langle D, R_{3} \rangle ,$ $\langle E, R_{2} \rangle ,$ $\langle F, R_{1} \rangle \}$,
        we can see that $\fFair(O_{1}) = 1$: $O_{1}$ is fair; in contrast, $\fFair(O_{2}) = 0$: $O_{2}$ is unfair.

    \end{example}

    The notion of fairness applied in the example does not consider the \emph{needs} of the agents. We present fairness measures that consider needs further below.

    \subsection{An Analysis of Equity and Equality}
    \label{subsec:measures}
    Let us demonstrate how the fairness metamodel can be applied to formalize and compare two well-known fairness measures: \emph{equality} and \emph{equity}.
    In the case of equality, every agent receives exactly the same amount of resources.
    \begin{definition}[Equality, Equity, and Strict Equity]
        \label{def:equality-equity}

        Let $\tFairScen = \langle \sAgent, \sResource, \sAgentAttribute, \sResourceAttribute \rangle$ be a fairness scenario, let $u \in \sResourceAttribute$ be a utility function, and $O$ an outcome.
        The \emph{accumulation of received resources} $r_{O} : \sAgent \to \sMeasure$ is:
        \begin{equation}
            r_{O}(a) = \sum _{\langle a , b \rangle \in O} u (b) \ .
        \end{equation}
        This function sums up the utility accumulated by an agent.

        \begin{itemize}
            \item The \emph{equality} fairness measure $\fFair_{eqa}$ is
            \[
                \fFair_{eqa} (O) = [\forall a, a' \in \sAgent \text{ it holds that } r_{O}(a) = r_{O}(a') ].
            \]

            \item If $q \in \sAgentAttribute$ is the need function, the \emph{equity} fairness measure $\fFair_{eqi}$ is
            \[
                \fFair_{eqi} (O) = [\forall a \in \sAgent \text{ it holds that } r_{O}(a) \geq q(a) ].
            \]

            \item The \emph{strict equity} fairness measure $\fFair_{seqi}$ is
            \[
                \fFair_{seqi} (O) = [\forall a \in \sAgent \text{ it holds that } r_{O}(a) = q(a) ].
            \]
        \end{itemize}
    \end{definition}
    Note in Definition~\ref{def:equality-equity}, equity stipulates that each agent receives at least as much as they need.
    One can strengthen the notion of \emph{equity} so that it is violated if an agent receives more than they need.
    We call this alternative notion of equity \emph{strict equity}; it constitutes a special case of equity.

    \begin{proposition}[Strict Equity Implies Equity]
        \label{lem:equity-strict-equity}
        For a fairness scenario $\tFairScen$, for every outcome $O$, the following implication holds:
        \[
            \fFair_{seqi}(O) = 1 \Rightarrow \fFair_{eqi}(O) = 1 .
        \]
    \end{proposition}
    \begin{proof}
        Let $\tFairScen = \langle \sAgent, \sResource, \sAgentAttribute, \sResourceAttribute \rangle$ be a fairness scenario, and let $q : \sAgent \to \sMeasure$ be the need function $q \in \sAgentAttribute$. If $\fFair_{seqi}(O) = 1$, then by definition $\forall a \in \sAgent \ (r_{O}(a) = q(a) ) $.
        Consequently, $\forall a \in \sAgent \ (r_{O}(a) \geq q(a) ) $ must hold as well, and therefore $\fFair_{eqi}(O) = 1$.
    \end{proof}
    More interestingly, we can show that equality can be reduced to strict equity by stipulating that the equally distributed available resources are sufficient to satisfy the needs. In other words, if identical amounts are distributed to every agent and each one requires the same amount, then it is possible to \emph{reduce} a fairness scenario to another one containing a need function to satisfy the needs of all the agents.
    As a prerequisite, we fix background sets of fairness scenarios $\cal F$ and outcomes $\cal O$.
    \begin{proposition}[Equality Reduced to Strict Equity]
        \label{lem:equality-strict-equity}
        For every fairness scenario $\tFairScen$ and outcome $O$, there exists a function $\tau : \cal F \times \cal O \to \cal F$, such that
        \[
            \fFair_{eqa}(O) = 1 \Longleftrightarrow \fFairTauSeqi(O) = 1.
        \]
    \end{proposition}

    \begin{proof}
        Let $\tFairScen = \langle \sAgent, \sResource, \sAgentAttribute, \sResourceAttribute \rangle$ be a fairness scenario and $O$ an outcome. Choose an arbitrary element $a_{0} \in \sAgent$, which is well-defined because $\sAgent$ is non-empty and define the function $q : \sAgent \to \sMeasure$ as follows: $q(a) := r_{O}(a_{0})$ (i.e., $q$ is fixed, independently of input $a$).
        Also, let $\tau(\tFairScen, O) = \langle \sAgent, \sAgentAttribute \cup \{q\} , \sResource, \sResourceAttribute \rangle$.

        ($\Rightarrow$) Assume that $\fFair_{eqa}(O) = 1$ and choose $a, a'  \in \sAgent$. Then, by definition it holds that $r_{O}(a) = r_{O}(a')$, specifically $r_{O}(a) = r_{O}(a_{0})$. By definition of $q$, it holds that $q(a) = r_{O}(a_{0})$. Since $a$ is arbitrary, we have $\forall a \in \sAgent$, $r_{O}(a) = q(a)$ and therefore $\fFairTauSeqi(O) = 1$.

        ($\Leftarrow$) Assume that $\fFairTauSeqi(O) = 1$ and choose $a \in \sAgent$. Then, by definition it holds that $r_{O}(a) = q(a) $. By definition of $q$, it holds that $q(a) = r_{O}(a_{0})$. Since $a$ is arbitrary, $\forall a, a' \in \sAgent$ we have $r_{O}(a) = r_{O}(a')$, and therefore $\fFair_{eqa}(O) = 1$.
    \end{proof}

    All agents need what an arbitrary agent receives. Only if all agents receive the same, then their needs are strictly met.

    \subsection{Preferences}
    \label{subsec:preferences}

    In the examples above, we use quantitative functions to evaluate fairness in the distribution of resources among agents. However, we can consider qualitative functions for that purpose as well. In fact, attributes can be used to determine preferences and thus define fairness measures. Note that qualitative functions can also be used to group agents into categories.

    \begin{definition}[Ordinal Preference Function]
        An \emph{ordinal preference function} $v$ is a function $v : \sAgent \to \fPerm(\sResource)$, where $\fPerm(\sResource) = \{ (b_{1}, b_{2}, \ldots, b_{n}) \mid \{b_{1}, b_{2}, \ldots, b_{n}\} = \sResource\}$, such that for each agent, it defines a ranking (a strict total preference order) over all resources, ordered from the most preferred to the least preferred.
        The notation $b_{1} \succ_{a} b_{2}$ indicates that $b_{1}$ precedes $b_{2}$ in $v(a)$, and is read as ``$a$ prefers $b_{1}$ over $b_{2}$''.
        Note that this particular modeling does not allow for \emph{ties}, i.e., situations in which an agent is indifferent between two resources. To accommodate ties, the definition of ranking can be relaxed to a weak order, where every pair is comparable but some pairs may be considered equally good.
    \end{definition}

    \begin{example}[Ordinal Preferences]
        \label{example:preferences}
        Assume that $\sAgent$ is a set of agents $\sAgent = \{A, B, C, D, E \}$, and each agent can receive a jacket that is \emph{large} ($L$) or \emph{small} ($S$), i.e., $\sResource = \{ L, S\}$. We allow each agent to choose different jackets in some order of preference. To model each agent's preference, we define the attribute $v : \sAgent \to \fPerm(\sResource)$,
        $v = \{ \langle A, (S, L) \rangle,$ $ \langle B, (S, L) \rangle,$ $
        \langle C, (L, S) \rangle,$ $ \langle D, (L, S) \rangle,$ $ \langle E, (L, S) \rangle \}$.

        In this case, B would be \emph{satisfied} with a large jacket but would \emph{prefer} a small one, C prefers a large jacket but would accept a small one. The fairness measure can be that every agent receives at least one of the resources in the preferences, which can be written as
        \[
            \fFair (O) = [\forall a \in \sAgent , \exists b \in \sResource \text{ s.t. } a \text{ receives } b].
        \]

        The agents can be satisfied with the following outcome:
        $O = \{ \langle A, S \rangle, $ $\langle B, L \rangle, $ $ \langle C, S \rangle, $ $ \langle D, L \rangle, $ $ \langle E, L \rangle \} $,
        but this does not prevent that C envies B's jacket and that B envies C's jacket. If they exchange their jackets:
        $O = \{ \langle A, S \rangle, $ $\langle B, S \rangle, $ $ \langle C, L \rangle, $ $ \langle D, L \rangle, $ $\langle E, L \rangle \}$,
        no agent envies another agent's jacket. The concept in which no agent envies the outcome of another agent is called \emph{envy-freeness}~\cite{Foley-1967-ResourceAllocation,Amanatidis-2018-ComparingEnvyFreeness,Richter.Rubinstein-2020-EconomicTheory,Li-2024-PriceEnvyFreeness}. We write a new fairness measure $\fFair (O)$ considering what we call a \emph{weak} envy-freeness, where each agent receives at least one resource that it most preferred than any other resource received by any other agent.

        \begin{itemize}
            \renewcommand\labelitemi{}
            \item $\sAgent$ = $\{A, B, C, D, E\}$,
            $\sResource$ = $\{ L, S \}$,
            \item $\sAgentAttribute = \{ v : \sAgent \to \fPerm(\sResource)$,
            $v = \{ \langle A, (S, L) \rangle, $ $ \langle B, (S, L) \rangle, $ $ \langle C, (L, S) \rangle,$
            $ \langle D, (L, S) \rangle, $ $ \langle E, (L, S) \rangle \} $,
            $\sResourceAttribute$ = $\emptyset$,
            \item $\fFair (O)$ = [ $\nexists a, a' \in \sAgent, b' \in \sResource$, s.t. $a'$ receives $b'$, $a$ does not receive $b'$,
            \ \ and $\forall b \in \sResource$ s.t. $a$ receives $b$ it holds that $b' \succ_{a} b$ ].
        \end{itemize}

        Here, $v$ represents the resource preference of each agent and no resource attribute is needed.
    \end{example}
    The change in outcomes in Example~\ref{example:preferences} is a \emph{Pareto improvement}, because at least one agent is better off without leaving anyone else worse off. An outcome is \emph{Pareto optimal} when no Pareto improvement can be applied.

    \subsection{Group and Individual Fairness}
    \label{subsec:group-and-individual}

    \emph{Group fairness} ensures that different demographic groups, which may have protected attributes, such as race and gender, receive similar outcomes. It focuses on statistical parity across groups.
    \emph{Individual fairness} ensures that similar individuals receive similar outcomes. It emphasizes consistency in treatment based on relevant features, regardless of group membership. Group fairness and individual fairness may seem to be in conflict~\cite{Binns-2020-ConflictIndividualGroupFairness}, but we can think of group and individual fairness as complementary.
    Let us assume that an attribute can be considered either \emph{relevant} or \emph{irrelevant} to determine the distribution of a resource. According to group fairness, if an attribute $p$ is irrelevant, the groups with attribute $p$ should receive the same amount as the groups without attribute $p$. According to individual fairness, if an attribute $q$ is relevant, the individuals with a similar value of attribute $q$ should be treated similarly.

    We illustrate this with the following example.

    \begin{example}[Group and Individual Fairness]
        Assume that a group of agents $\sAgent = \{A, B, C, D, E, F \}$ apply for a loan ($L$), and that there is a protected demographic attribute $p$, which should be irrelevant for the loan application. Assuming that only $D$, $E$, and $F$ have that attribute creates two demographic groups: $G _{\lnot p} = \{ A, B, C\}$ and $G _{p} = \{D, E, F \}$.

        Following group fairness, those belonging to different demographic groups should be treated similarly, regardless of the group to which they belong. Assuming that 2 out of 3 loan applications are accepted, that relation should appear in $G_{\lnot p}$ and also in $G_{p}$. For example, if the applications of $A$, $B$, $E$, and $F$ are accepted and the rest ($C$ and $D$) rejected, group fairness holds in this case.

        Following individual fairness, if two applicants have nearly identical values for a critically relevant attribute, they should receive similar treatment. Assume that $q$ is an essential attribute to determine whether to give a loan, and that only $B$, $C$, $E$, and $F$ have this attribute. If $D$ gets the loan and $F$ does not, individual fairness is not observed in this outcome.

        We formalize this example as follows:

        \begin{itemize}
            \renewcommand\labelitemi{}
            \item $\sAgent$ = $\{A, B, C, D, E, F \}$,
            $\sResource$ = $\{ L \}$,
            \item $\sAgentAttribute$ = $ \{ p : \sAgent \to \sBoolean , p = \{ \langle A , false \rangle,$ $ \langle B , false \rangle,$ $ \langle C , false \rangle, $ $\langle D , true \rangle,$ $ \langle E , true \rangle, $ $\langle F , true \rangle \}, $ \ \ $q : \sAgent \to \sBoolean $, $q = \{ \langle A , false \rangle, $ $\langle B , true \rangle, $ $\langle C , true \rangle, $ $\langle D , false \rangle,$ $ \langle E , true \rangle, $ $\langle F , true \rangle \} \ \}$,
            $\sResourceAttribute $ = $\emptyset$, \\
            $\varepsilon = 10^{-2},$
            $\cdot \simeq_{\varepsilon} \cdot : \sMeasure \times \sMeasure \to \sBoolean$, \\
            $a \simeq_{\varepsilon} b =
            \begin{cases}
                true, & \text{if } a = b \text{ or } (a \neq b \text{ and } \dfrac{|a - b|}{\max(|a|, |b|)} < \varepsilon) \\
                false, & \text{otherwise}
            \end{cases}$
            \item $r_{p^{+}} =  \dfrac{| \{ a \in \sAgent \mid p(a) \text { and } a \text{ receives } L \} |}{ | \{a \in \sAgent \mid p(a) \} | } $,
            \item $r_{p^{-}} =  \dfrac{| \{ a \in \sAgent \mid \lnot p(a) \text { and } a \text{ receives } L \} |}{ | \{a \in \sAgent \mid \lnot p(a) \} | } $,
            \item $G \fFair (O)$ = $[r_{p^{+}} \simeq_{\varepsilon} r_{p^{-}} ]$ \ ,
            \item $I \fFair (O)$ = {\tiny{1:}} $[  \forall a, a' \in \sAgent , a \neq a' \text{ and } q(a) = q(a') \Rightarrow $ \\
            {\tiny{2:}} \ \ \ \ $( a \text{ receives } L \text { and } a' \text{ receives } L ) \text { or }$ \\
            {\tiny{3:}} \ \ \ \ $( a \text{ does not receive } L \text { and } a' \text{ does not receive } L ) ] $ \ .
        \end{itemize}

        Here, $p$ is an irrelevant protected attribute for group fairness, $q$ is an essential attribute for individual fairness, $\cdot \simeq_{\varepsilon} \cdot $ determines whether two quantities are similar up to a value $\varepsilon$, $r_{p^{+}}$ and $r_{p^{-}}$ are the ratios between the number of agents
        with/without the protected attribute $p$ that receive the loan compared to all that have/do not have $p$ respectively.

        We provide two different fairness measures: $G \fFair (O)$ for group fairness and $I \fFair (O)$ for individual fairness.
        Group fairness compares the relation between the number of agents with attribute $p$ that receive the loan compared to all the agents having $p$ is similar to the relation between those who receive the loan for the other group.
        Individual fairness requires that:
        \begin{enumerate}
            \item [1.] for every two different agents with the same value of the essential attribute $q$,
            \item [2.] either both receive the loan,
            \item [3.] or neither receives the loan.
        \end{enumerate}
    \end{example}

    \subsection{Continuous Fairness Measures}
    \label{subsec:continuous-fairness-measures}

    In Definition~\ref{def:fairness-measure}, we define $\fFair$ to range in the continuous interval $\sZeroOneInterval$, but we only have presented discrete examples up to this point. However, we can model a continuous example such as Jain's fairness index~\cite{Jain-1984-QuantitativeMeasureOfFairness}, which is a quantitative measure for assessing how evenly a resource is allocated among $n$ agents.

    \begin{definition}[Jain's Fairness Index]
        \label{def:jain-index}
        Given a set of agents indexed from 1 to $n$, such that each agent receives $x_{1}, x_{2},\ldots ,x_{n}$ respectively, the index is defined (on the left) and rewritten (on the right) as:
        \begin{equation}
            \label{eq:jain-equation}
            \begin{tabular}{cc}
                $
                J(x_{1}, \ldots , x_{n}) = \frac{\bigl(\sum \limits _{i=1}^{n} x_{i} \bigr)^2}{n \cdot \sum \limits_{i=1}^{n} x_{i}^{2}} \ ,

                $ &
                $
                \fFair(O) = \frac{\bigl(\sum \limits _{a \in \sAgent} r_{O}(a) \bigr)^2}{|\sAgent| \cdot \sum \limits_{a \in \sAgent} r_{O}(a)^{2}} \
                $. \\
            \end{tabular}
        \end{equation}

    \end{definition}

    Jain's fairness index is used to measure fairness in network resource allocation, to evaluate load balancing schemes in distributed systems, and to balance throughput in congestion control protocols. In particular, the index provides a single scalar score that can be used to compare different allocation strategies and to tune parameters for the desired level of fairness.

    In the following example, we apply this fairness measure.

    \begin{example}[Continuous Fairness Measure]

        Let us assume that agents need to access resources that represent different bandwidth values in a computer network.

        \begin{itemize}
            \renewcommand\labelitemi{}
            \item $\sAgent$ = $\{A, B, C, D \}$,
            $\sResource$ = $\{ M_{0}, M_{10} , M_{20} , M_{50} \}$,
            $\sAgentAttribute $ = $\emptyset$,
            \item $\sResourceAttribute$ = $ \{ u : \sResource \to \sMeasure , u = \{ \langle M_{0} , 0 \rangle, $ $ \langle M_{10} , 10 \rangle, $ $ \langle M_{20} , 20 \rangle, $ $ \langle M_{50} , 50 \rangle \} \ \}, $
            \item $\fFair(O)$ = as in Equation~(\ref{eq:jain-equation})
        \end{itemize}

        In this case, $u$ represents the utility in megabits per second (Mbps) of each resource.

        Considering the following outcomes:
        $O_{1} = \{ \langle A, M_{20} \rangle, $ $ \langle B, M_{20} \rangle, $ $ \langle C, M_{20} \rangle,$ $ \langle D, M_{20} \rangle \}$,
        $O_{2} = \{ \langle A, M_{20} \rangle, $ $ \langle B, M_{20} \rangle, $ $\langle C, M_{20} \rangle,$ $ \langle D, M_{0} \rangle \} $, and
        $O_{3} = \{ \langle A, M_{0} \rangle, $ $\langle B, M_{0} \rangle, $ $\langle C, M_{0} \rangle, $ $\langle D, M_{10} \rangle \}$,
        we obtain that
        $\fFair(O_{1}) = 1$, $\fFair(O_{2}) = 0.75$, and $\fFair(O_{3}) = 0.25$. We could interpret in words that $O_{1}$ is perfectly fair, $O_{2}$ is moderately fair, and $O_{3}$ is clearly unfair, but in practice the nuances of how fair these values are strongly dependent on context.
    \end{example}
    A continuous fairness measure that is more closely aligned with \emph{social} (in contrast to \emph{technical}) applications is the Gini index, which we cover in the example below.
    \begin{example}[Complement of the Gini Index]
        The \emph{Gini index} is a statistical measure of inequality often used to quantify the inequality of income or wealth within a population.
        The index has also been applied in other contexts, such as decision tree algorithms in machine learning.
        In the interpretation of the Gini index, the closer the Gini index is to 0, the more equal the distribution, and the closer to 1, the more unequal. This works exactly opposite to Jain's index and our definition for a fairness measure. Because of this, we define the \emph{complement of the Gini index} by inverting the output of the Gini index. The Gini index is defined (on the left) and rewritten (on the right) as:

        \begin{equation}
            \label{def:gini-equation}
            \begin{tabular}{cc}
                $ G = \displaystyle\frac{\displaystyle\sum_{i=1}^{n} \displaystyle\sum_{j=1}^{n} |x_{i} - x_{j}|}{2 \cdot n \cdot \displaystyle\sum_{i=1}^{n} x_{i}} \ ,
                $ &
                $\fFair (O) = 1 - \displaystyle\frac{\displaystyle\sum_{a_{1} \in \sAgent} \displaystyle\sum_{a_{2} \in \sAgent} |r_{O}(a_{1}) - r_{O}(a_{2}) |}{2 \cdot |\sAgent| \cdot \displaystyle\sum_{a \in \sAgent} r_{O}(a) }$ \\
            \end{tabular}
        \end{equation}

        where $n$ is the number of agents (or households) and
        $x_{i}$ is the income (or wealth) of agent $i$.

        As in the case of Jain's index in Definition~\ref{def:jain-index}, we can plug in the agents of set $\sAgent$, such that $n$ is $|\sAgent|$, and $r_{O}(a)$ is the amount received by each agent $a$.
        The agents are the households, and the resource is the wealth to be distributed.

        \begin{itemize}
            \renewcommand\labelitemi{}
            \item $\sAgent$ = $\{A, B, C, D, E, F \}$, $\sResource$ = $\{ R_{5}, R_{10}, R_{15}, R_{20}, R_{50}, R_{100} \}$,
            \item $\sAgentAttribute $ = $\emptyset$, $\sResourceAttribute$ = $ \{ u : \sResource \to \sMeasure , u = \{ \langle R_{5} , 5 \rangle, \langle R_{10} , 10 \rangle, \langle R_{15} , 15 \rangle, \langle R_{20} , 20 \rangle, $
            \ \ $\langle R_{50} , 50 \rangle, \langle R_{100} , 100 \rangle \} \ \}, $
            \item $\fFair(O)$ = as in Equation~(\ref{def:gini-equation})

        \end{itemize}

        Here, $u$ is the utility (for example, in thousands of euros) of each resource.

        Considering the following outcomes:
        $O_{1} = \{ \langle A, R_{20} \rangle, $ $ \langle B, R_{20} \rangle, $ $\langle C, R_{20} \rangle, $ $\langle D, R_{20} \rangle, $ $\langle E, R_{20} \rangle, $ $\langle F, R_{20} \rangle  \}$,
        $O_{2} = \{ \langle A, R_{5} \rangle, $ $\langle B, R_{10} \rangle, $ $\langle C, R_{15} \rangle, $ $\langle D, R_{20} \rangle, $ $\langle E, R_{50} \rangle, $ $\langle F, R_{100} \rangle  \}$,
        $O_{3} = \{ \langle A, R_{5} \rangle, $ $\langle B, R_{10} \rangle \}$,
        we obtain that
        $\fFair(O_{1}) = 1 - \frac{0}{1440} = 1$, $\fFair(O_{2}) = 1 - \frac{1200}{2400} = 0.5$, and $\fFair(O_{3}) = 1 - \frac{130}{180} \approx 0.28$. We could interpret in words that $O_{1}$ is perfectly fair, $O_{2}$ is rather unfair, and $O_{3}$ is clearly unfair.
    \end{example}

    \begin{example}[Fairness Measure for Equalized Odds]
        The COMPAS (Correctional Offender Management Profiling for Alternative Sanctions)
        system~\cite{ProPublica-2016}
        is a well-known example of a risk assessment tool used in the U.S. criminal justice system to predict the likelihood of recidivism (re-offending after a prior arrest).
        Studies, most notably by ProPublica in 2016\footnote{\url{https://github.com/propublica/compas-analysis}},
        found that COMPAS exhibited racial bias: it was more likely to falsely flag Black defendants as future criminals (higher false positive rate) and more likely to misclassify White defendants as low risk (higher false negative rate).

        Based on that example, we designed a fairness measure to detect bias in a prediction system. The system gives two possible scores (resources) with their analogous score in COMPAS:
        \begin{itemize}
            \item $R_{\mathrm{low}}$ : Low Risk (scores 1--4)
            \item $R_{\mathrm{high}}$ : Medium Risk (scores 5--7) and High Risk (scores 8--10)
        \end{itemize}
        The prediction system is defined as follows.

        \begin{itemize}
            \renewcommand\labelitemi{}
            \item $\sAgent$ = $\{A, B, C, D, E, F \}$,
            $\sResource$ = $\{ R_{\mathrm{low}}, R_{\mathrm{high}} \}$,
            \item $\sAgentAttribute$ = $ \{ p : \sAgent \to \sBoolean , \ p = \{ \langle A , false \rangle, $ $ \langle B , false \rangle, $ $ \langle C , false \rangle, $  $\langle D , true \rangle, $ $ \langle E , true \rangle, $ $ \langle F , true \rangle \}, $
            $res : \sAgent \to \sResource , \ res = \{\langle A , R_{\mathrm{low}} \rangle, $ $ \langle B , R_{\mathrm{high}} \rangle, $ $ \langle C , R_{\mathrm{high}} \rangle, $
            $\langle D , R_{\mathrm{low}} \rangle, $ $ \langle E , R_{\mathrm{low}} \rangle, $ $ \langle F , R_{\mathrm{high}} \rangle \}, $ \\
            $\sResourceAttribute$ = $ \{ u : \sResource \to \sMeasure , u = \{ \langle R_{\mathrm{low}} , 0 \rangle, \langle R_{\mathrm{high}} , 1 \rangle \} \ \}, $ \\
            $ \overline{(x_{i})_{i=1}^{n}} \text{ is the arithmetic mean of the sequence } (x_{1}, x_{2}, \ldots , x_{n}) $                                                                                                      \\
            $corr : \displaystyle\bigcup_{n=1}^{\infty} (\sMeasure \times \sMeasure) \to [-1, 1],$      \\
            $
            corr ((x_{i})_{i=1}^{n}, (y_{i})_{i=1}^{n}) = \frac{\displaystyle\sum_{k=1}^{n} (x_{k} - \overline{(x_{i})_{i=1}^{n}}) \, (y_{k} - \overline{(y_{i})_{i=1}^{n}})}
            {\sqrt{\displaystyle\sum_{k=1}^{n} (x_{k} - \overline{(x_{i})_{i=1}^{n}})^2} \;
            \sqrt{\displaystyle\sum_{k=1}^{n} (y_{k} - \overline{(y_{i})_{i=1}^{n}})^2}}
            $

            \item $\fFair(O)$ =    $| corr(\ (\,[\langle x_{i}, R_{\mathrm{high}} \rangle \in O \,\land\, res(x_{i}) \neq R_{\mathrm{high}}]\,)_{i=1}^{n}, (\,[p(x_i)]\,)_{i=1}^{n} \ ) |$
        \end{itemize}

        Here, $p$ is the protected attribute, such as skin color in COMPAS, $res$ is the ground truth based on facts, and $u$ is the associated risk of each score. We use the Pearson correlation coefficient, but another correlation coefficient can also be used. Recall that the square brackets $[$ $]$ follow the Iverson notation and that the fairness measure should only return values in $[0, 1]$.
    \end{example}

    \section{Operationalization}
    \label{sec:operationalization}

    In this section, we show how we operationalize the concepts presented above, that is, how we translate abstract ideas into concrete, formal representations that can be implemented and measured.
    Essentially, an instance of the \AR metamodel provides a formal representation of a fairness scenario
    $\tFairScen = \langle \sAgent, \sResource, \sAgentAttribute, \sResourceAttribute \rangle$. Assuming that the agent and resource attributes are given, the fairness measure still needs to be defined.
    As we can see from the examples above, defining a fairness measure can be both complex and error-prone.
    To address this, we use the \Tiles framework~\cite{Mendez.Kampik.Aler.Dignum-2024-SCAI} as a way of defining the fairness measure for a given fairness scenario. The purpose of this framework is to improve readability and reliability of the fairness measures. To achieve this, we decompose the fairness measure into smaller pieces of composable blocks.

    \subsection{Design of the Blocks}
    \label{subsec:design}

    The design of blocks in \Tiles is based on the concept of a \emph{module} in software engineering.
    Recall that a fairness scenario is a tuple $\tFairScen = \langle \sAgent, \sResource, \sAgentAttribute, \sResourceAttribute \rangle$. We assume that $\sResourceAttribute$ contains a function $u : \sResource \to \sMeasure$ \emph{(utility)}, $\sAgentAttribute$ contains a function $q : \sAgent \to \sMeasure$ \emph{(needs)}, and we define the auxiliary function $r_{O} : \sAgent \to \sMeasure$ \emph{(accumulates)} as in Definition~\ref{def:equality-equity}, which depends on $u$.

    We represent the agents $\sAgent$ as a sorted sequence of identifiers $\fOrderedSeq{A}$, denoted by \tilesfun{all-agent}. We define \tilesfun{accumulates} and \tilesfun{needs} by applying $r_{O}$ and $q$, respectively, to each element of the sequence. Additionally, we define \tilesfun{all-equal} as the result of checking whether all the elements in the sequence are equal, and \tilesfun{all-at-least} as the result of checking, for each pair in the sequence, if its first component is greater than or equal to the second component.

    We define equality with the following pipeline:

    \begin{equation}
        \fFair_{eqa} (O) = \tilesfun{all-equal}(\tilesfun{accumulates}(\tilesfun{all-agent}))
        \label{eq:equality-pipeline}
    \end{equation}

    $\fFair_{eqa}$ in Definition~\ref{def:equality-equity} is equivalent to $\fFair_{eqa}$ given in Equation~\ref{eq:equality-pipeline}.

    \begin{proposition}[Correctness of the Equality Pipeline]
        Let $\tFairScen = \langle \sAgent, \sResource, \sAgentAttribute, \sResourceAttribute \rangle$ be a fairness scenario, $\fFair_{eqa}$ defined as in Definition~\ref{def:equality-equity}, and $\fFair_{eqa'}$ defined as in $(\ref{eq:equality-pipeline})$. Then, for every outcome $O$,
        \[
            \fFair_{eqa} (O) = 1 \Longleftrightarrow \fFair_{eqa'} (O) = 1 \ .
        \]
    \end{proposition}

    \begin{proof}
        Applying the expansion of $\fFair_{eqa} (O)$ in (\ref{eq:equality-pipeline}) yields:
        \begin{equation*}
        [\forall m, m' \in (r_{O}(a))_{a \in \fOrderedSeq{\sAgent}} \ (m = m') \ ]
            .
        \end{equation*}

        Notice that $\sAgent$ is non-empty and finite, and that $a \in \sAgent$ if and only if $a \in \fOrderedSeq{\sAgent}$.
        In particular, the sequence $(r_{O}(a))_{a \in \fOrderedSeq{\sAgent}}$ contains, for each agent $a$, the value $r_{O}(a)$.
        In other words, for each value $m$, $m \in (r_{O}(a))_{a \in \fOrderedSeq{\sAgent}}$ if and only if $m \in \{ r_{O}(a) \mid a \in \sAgent \}$.

        Then, $\fFair_{eqa'} (O)$ holds if and only if (by definition)
        \[
            \forall m, m' \in (r_{O}(a))_{a \in \fOrderedSeq{\sAgent}} \ (m = m') \ .
        \]
        In turn, $\fFair_{eqa'} (O)$ holds if and only if
        \begin{equation}
            \forall m, m' \in \{ r_{O}(a) \mid a \in \sAgent \} \ (m = m') \ .
            \label{eq:equality-of-values}
        \end{equation}
        Since $m \in \{ r_{O}(a) \mid a \in \sAgent \}$ if and only if there is an $a \in \sAgent$, such that $m = r_{O}(a)$, we can re-write (\ref{eq:equality-of-values}) as
        \[
            \forall a, a' \in \sAgent \ (r_{O}(a) = r_{O}(a')),
        \]
        which is the definition of $\fFair_{eqa} (O)$.
    \end{proof}
    We define equity with the following pipeline

    \begin{equation}
        \fFair_{eqi} (O) = \tilesfun{all-at-least}(\tilesfun{accumulates}(\tilesfun{all-agent}), \tilesfun{needs}(\tilesfun{all-agent}))
        \label{eq:equity-pipeline}
    \end{equation}

    Observe that $\fFair_{eqi}$ in Definition~\ref{def:equality-equity} is equivalent to $\fFair_{eqi}$ given in $(\ref{eq:equity-pipeline})$.

    \begin{proposition}[Correctness of the Equity Pipeline]
        Let $\tFairScen = \langle \sAgent, \sResource, \sAgentAttribute, \sResourceAttribute \rangle$ be a fairness scenario, $\fFair_{eqi}$ defined as in Definition~\ref{def:equality-equity}, $\fFair_{eqi'}$ defined as in $(\ref{eq:equity-pipeline})$. Then, for every outcome $O$,
        \[
            \fFair_{eqi} (O) = 1 \Longleftrightarrow \fFair_{eqi'} (O) = 1 \ .
        \]
    \end{proposition}

    \begin{proof}
        For the case of equity, consider:
        $X = (r_{O}(a))_{a \in \fOrderedSeq{\sAgent}}$ and
        $Y = (q(a))_{a \in \fOrderedSeq{\sAgent}}$.

        Notice that, by construction, $X$ and $Y$ have as many elements as $\sAgent$. Considering that $|X| = |Y|$, $\fFair_{eqi'} (O)$ can be expanded as:
        \[
            \makecell[l] {[ \forall i \in \sNat, 1 \leq i \leq |X| \Rightarrow \ \\
            \ \ \text{ the } i \text{-th element of } (r_{O}(a))_{a \in \fOrderedSeq{\sAgent}} \geq \ \text{ the } i \text{-th element of } (q(a))_{a \in \fOrderedSeq{\sAgent}} ] \ .
            }
        \]
        This is 1 if and only if the following holds:
        $\forall a \in \sAgent \ (r_{O}(a) \geq q(a))$,
        which is how $\fFair_{eqi} (O)$ is defined in Definition~\ref{def:equality-equity}.
    \end{proof}

    The functional notation for $\Tiles$ can be represented using a graphical notation.

    \subsection{Graphical Notation}
    \label{subsec:notation}

    One of the key aspects of \Tiles is the ability to represent configurations in a graphical notation that clarifies how a configuration works. The whole configuration is a formal representation of how the blocks are interconnected. Each block, called \emph{tile}, can connect with others to produce a specific definition.

    \begin{concept}[Tile]
        \label{def:tile}
        A \emph{tile} is a construct that contains a name, a function, an input type, an output type, and contextual information, which includes the fairness scenario, constants, and auxiliary functions. The tile is represented as
        \[
            \tiles{_{\alpha} \ name \ _{\beta}}
        \]
        where $\alpha$ and $\beta$ are type annotations for the input type and the output type respectively, and $\tilesfun{name}$ is the function name. The input type is omitted if the tile represents a constant.
        A type in \Tiles can be
        \begin{itemize}
            \item atomic: \tilestype{a} (agent), \tilestype{r} (resource), \tilestype{m} (quantity or measure), and \tilestype{b} (Boolean);
            \item a tuple composed of other types: $\langle \alpha_{1}, \ldots, \alpha_{n} \rangle$; or
            \item a sequence of a type: $(\alpha)$.
        \end{itemize}

    \end{concept}

    The type annotation in \Tiles can be used not only to specify the connection between the output of a tile and the input of another type, but also to denote an input variable name when the tile has parameters.
    Parametric tiles are also useful for defining customized tiles.

    \subsection{Implementation}

    The \Tiles framework is implemented as an open-source project\footnote{\url{https://github.com/julianmendez/tiles}} written in the \Soda language~\cite{Mendez-2023-Soda}, an open-source functional language\footnote{\url{https://github.com/julianmendez/soda}}. The code in \Soda can be formally proved using the Lean~\cite{Lean-2013} proof assistant, and seamlessly integrated with the Java Virtual Machine ecosystem, allowing efficient execution~\cite{Mendez.Kampik-2025-EUMAS}.
    The \Tiles implementation written in \Soda aims to follow the notation used for the pipelines as accurately as possible, where each tile of the implementation has its source code directly accessible.

    The framework includes detailed components that ensure the correct construction of pipelines. These components are particularly focused on zipping and unzipping sequences, as well as creating and projecting tuples. Since \Soda is statically typed and \Tiles is a typed framework, the type consistency of the entire pipeline can be verified at compile time.

    \section{Discussion}
    \label{sec:discussion}

    This section informally discusses the capabilities and limitations of the \AR metamodel and the \Tiles framework. The framework is designed around the concept of \emph{flow}, where data moves through a pipeline. This pipeline connects tiles, forming a directed acyclic graph with a single start and end point. We do not provide an effective algorithm for constructing pipelines from the first-order formula given by the fairness measure, since it is not possible in the general case.

    \subsection{Expressiveness}

    The \AR metamodel can model a wide range of fairness scenarios, though not all possible ones. We specifically focus on scenarios that involve attributes of agents and resources. In particular, the metamodel supports the combination of multiple agents and resources, as demonstrated in Example~\ref{example:preferences}, where an agent or resource serves as the first parameter of a function that returns a sequence.

    \subsection{Decidability and Time Complexity}
    Pipelines built within the \Tiles framework are decidable provided that their tiles are decidable. When contracts between tiles are respected, undefined values cannot arise.

    With respect to time complexity, the pipeline structure guarantees the absence of loops in the diagram, ensuring that no tile is evaluated more than once. The worst-case complexity of a pipeline is determined by the maximum complexity of its individual tiles.

    \subsection{Applicability}

    The \Tiles framework has broad applicability beyond the fairness domain. As a software engineering tool, it can handle any scenario that involves finite, iterable sets of identifiers and attributes. The type system in \Tiles is flexible and includes sets of identifiers, quantities, Boolean values, tuples, and sequences.

    That said, the \Tiles framework is not intended to be applied at all levels of a scenario. The framework describes connections and processes to provide a better explanation of complex formulas. However, if the notation of the formulas is clear enough, they should be used instead of the pipelines.

    The \Tiles framework is primarily intended for modeling, but the generated pipelines are naturally simple to execute in parallel. This is because pipelines are usually designed to process several agents, resources, or quantities at the same time.
    However, this requires expertise in modeling these pipelines, since parallel execution can be effortful to design.

    \subsection{Limitations}

    A clear limitation of our approach is the generic nature of the metamodel. To address this, we introduce \Tiles as a layer of modular building blocks. Although \Tiles is conceptually elegant, it may pose, similar to other declarative notations, challenges in terms of readability.
    Studying (and potentially improving) the readability of \Tiles can therefore be considered important future work.
    For example, one could instantiate fairness measures and scenarios in several languages and systematically compare understandability, or one could conduct perceived usefulness studies analogous to~\cite{DBLP:journals/sosym/Jalali23} (which empirically compares several declarative modeling languages in this regard).

    Furthermore, this work does not address modeling \emph{specific} real-world problems: we either use \emph{toy examples} that facilitate a better understanding, or work with generally applicable measures, e.g., to showcase broad applicability of the conceptual level, or to make fundamental observations about their relationship (as in the case of strict equity vs. equality).
    Future work could evaluate our metamodel and its operationalization in specific case studies, in which real-world users are faced with fairness modeling and analysis challenges.

    \section{Conclusion}
    \label{sec:conclusion}

    In this paper, we have introduced a metamodel of fairness and demonstrated its application through concrete examples. The metamodel offers a structured approach to understanding and evaluating fairness in various scenarios. It begins by identifying agents, resources, and their relevant attributes, and concludes by the construction of complex fairness measures.

    We have provided examples of an application of the metamodel of central fairness concepts, such as equality, equity, group fairness, and individual fairness. We have explored the use of the metamodel spanning from economics and game theory, like preferences and envy-freeness, to a concept vastly applied in computer networks, like Jain's fairness index.

    In addition to the multiple examples, we have provided a method for operationalizing the metamodel, which simplifies the formal notation.
    The graphical and conceptual representation is supported by a modular design. This formal and practical notation for defining functions can also be applied across various domains.
    Moreover, we provide an implementation of this representation.

    Looking ahead, the fairness framework can be extended and revised when exploring its use in real-world scenarios.
    In this context, it is essential to provide further tooling that simplifies practical modeling.

    \subsubsection*{Acknowledgements}
    This work was partially supported by the Wallenberg AI, Autonomous Systems and Software Program (WASP) funded by the Knut and Alice Wallenberg Foundation.

    \bibliographystyle{LNGAI}
    \bibliography{main}

\end{document}